\documentclass[]{article}

\usepackage{authblk}

\usepackage{graphics} 
\usepackage{epsfig} 
\usepackage{amsmath} 
\usepackage{amssymb}  
\usepackage{amsthm}	

\usepackage{color}
\usepackage{capt-of}
\usepackage{mathtools}

\newcommand{\rmm}[1]{\left< \mathrm{#1} \right>}
\newcommand{\R}{\mathbb{R}}

\newtheorem{rem}{Remark}
\newtheorem{theorem}{Theorem}
\newtheorem{prop}{Proposition}
\newtheorem{corollary}{Corollary}

\newtheorem{definition}{Definition}[section]
\newtheorem{prob}{Problem}[section]

\DeclareMathOperator*{\argmin}{arg\,min}

\begin{document}

\title{Formal Synthesis of Analytic Controllers for Sampled-Data Systems via Genetic Programming.} 

\author{Cees F. Verdier and Manuel Mazo Jr.} 
\date{\vspace{-5ex}}

\affil{Delft Center for Systems and Control, Delft University of Technology, The Netherlands (e-mail: c.f.verdier@tudelft.nl, m.mazo@tudelft.nl)}
\maketitle

\begin{center}
\thanks{Supported by NWO Domain TTW under the CADUSY project \#13852. }

\thanks{The original version of this article has been accepted to CDC 2018. This version contains minor corrections. }
\end{center}

\begin{abstract}

This paper presents an automatic formal controller synthesis method for nonlinear sampled-data systems with safety and reachability specifications. Fundamentally, the presented method is not restricted to polynomial systems and controllers. We consider periodically switched controllers based on a Control Lyapunov Barrier-like functions. The proposed method utilizes genetic programming to synthesize these functions as well as the controller modes. Correctness of the controller are subsequently verified by means of a Satisfiability Modulo Theories solver. Effectiveness of the proposed methodology is demonstrated on multiple systems.
\end{abstract}

\section{Introduction}
Modern controller design for nonlinear continuous systems often involves both reachability and safety specifications. Furthermore, digital controller implementations typically impose that states are measured periodically and that control signals are held constant in between sampling. This paper proposes an approach to automatically synthesize periodically switched state feedback controllers for a special subclass of safety and reachability specifications for nonlinear sampled-data systems. 

Two popular paradigms for automatic controller synthesis for reachability and safety specifications are: 1) abstraction and simulation, and 2) Control Lyapunov functions (CLF) and Control Barrier Functions (CBF). 

The first approach abstracts the infinite system to a finite one, which simplifies the formal controller synthesis for temporal logic specifications \cite{Tabuada2009}. For nonlinear systems, tools implementing this approach include PESSOA \cite{Mazo2010}, SCOTS \cite{Rungger2016} and CoSyMa \cite{Mouelhi2013}. The second approach deals with the system as an infinite system. Control Lyapunov functions \cite{Arstein1983} and Control Barrier Functions \cite{Wieland2007} are design tools for stabilization and safety specifications respectively. In \cite{Romdlony2016} and \cite{Ames2016} attempts are made to combine both CLFs and CBFs.
Automatic synthesis of these functions is often done by posing the problem as a sum of squares (SOS) problem, which can be solved through convex optimization, see e.g. \cite{Papachristodoulou2002} and \cite{Tan2004}. Drawbacks of the abstraction and (bi-)simulation approach are that it requires discretization of the state space and that the resulting controller is often an enormous look-up table in the form of a sparse matrix or a binary decision diagram (BDD). On the other hand, the SOS programming paradigm is limited to  polynomial systems. Although reformulation of some nonpolynomial systems to an SOS formulation exists, e.g. \cite{Papachristodoulou2002,Hancock2013} and references therein, polynomial Lyapunov functions can be too restrictive, as global asymptotical stability of a polynomial system does not imply the existence of a polynomial Lyapunov function \cite{Ahmadi2011}.

To overcome these limitations, we propose a framework which uses genetic programming (GP) in combination with a Satisfiability Modulo Theories (SMT) solver. GP is an evolutionary algorithm which evolves encoded representations of symbolic functions \cite{Koza1992}, rather than just fitting optimal parameters given a predefined structure. An SMT solver is a tool which uses a combination of background theories to determine whether a first-order logic formula can be satisfied \cite{Barrett2009}. Our approach uses a CLF-like function and a predefined switching law that infers a reachability and safety specification. The proposed framework uses GP to automatically generate both candidate CLFs and optionally the controller modes of a periodically switched state feedback controller. The SMT solver is subsequently used to formally verify the candidate solutions. By using GP, we allow ourselves to search for solutions that include nonpolynomial functions. Furthermore, the synthesized controllers are expressed as analytic expressions that are significantly more compact than BDDs returned by abstraction-based methods.

This work is a follow-up to \cite{Verdier2017}, in which also a combination of GP and SMT solvers is used. The main contributions of this paper are: 1) synthesis w.r.t. a predefined periodic sampling time, rather than arbitrary switching with a (more conservative) minimum dwell-time and 2) the use of a different and less conservative CLBF. Additionally, more benchmark examples are provided. Other related work is found in \cite{ravanbakhsh2016}, in which robust CLFs for switched systems with reach-while-stay (RWS) specifications are synthesized using a counterexample-guided synthesis. However, in \cite{ravanbakhsh2016} the controller modes are pre-specified, while in our approach these modes can also be discovered automatically, eliminating the need for prior input space discretization. Furthermore, this paper  extends the set of specifications to include invariance of the goal set. Finally, similar to \cite{Verdier2017}, the theoretical lower bounds on the minimum dwell-times reported in \cite{ravanbakhsh2016} are often very conservative.

\paragraph{Notation}
Let $\mathbb{Z}_{\geq 0} = \left\{ 0,1,2 \dots \right\}$.  Let us denote the boundary and the interior  of a set $D$ with $\partial D$ and $int(D)$ respectively. The image and inverse image of set $A$ under $f$ are denoted by $f[A]$ and $f^{-1}[A]$. Finally, the Euclidean norm is denoted by $\| \cdot \|$.

\section{Problem definition}

In this paper we design sampled-data state feedback controllers for nonlinear continuous-time systems described by
\begin{equation}
\label{eq:system}
\dot{\xi}(t) = f(\xi(t),u(t)),
\end{equation}
where the variables $\xi(t) \in \mathcal{X} \subseteq \R^n$ and $u(t) \in \mathcal{U} \subseteq \R^m$ denote the state and input respectively. Due to the sampled-data nature of the controller, $u(t) = g(x(t_k))$, $\forall t \in [t_k, t_k+h)$, where $h>0 $ denotes a constant sampling time.

\subsection{Control specification}
Given a compact safe set $S \subseteq  \mathcal{X}$, compact initial set $I \subset S$ and compact goal set $G  \subset S$, we consider the following specifications:
\begin{enumerate}
\item[$\mathtt{CS}_1$] \textit{Reach while stay (RWS)}: all trajectories starting in $I$ eventually reach  $G$, while staying within $S$:
\end{enumerate}
\begin{equation}
\begin{array}{r}
\forall \xi(t_0) \in I,\exists T,\forall t \in [t_0, T]:
 \xi(t) \in S  \wedge \xi(T) \in G.
\end{array}
\label{eq:RWS}
\end{equation}
\begin{enumerate}
\item[$\mathtt{CS}_2$] \textit{Reach and stay while stay (RSWS)}: all trajectories starting from $I$ eventually reach and stay in $G$, while staying within $S$:
\end{enumerate}
\begin{equation}
\begin{array}{r}
\!\! \! \forall \xi(t_0) \in I,\exists T ,\forall t \geq t_0,\forall \tau \geq T \! :  \! 
 \xi(t) \in S \wedge \xi(\tau) \in \! G. \! \! \!
 \end{array}
 \label{eq:RSWS}
\end{equation}

This paper addresses the following problem:
\begin{prob}
\label{prob:1}
Given the compact sets $(S,I,G)$ and system \eqref{eq:system}, synthesize a sampled-data state feedback controller $u(t) = g(x(t_k))$ such that the closed-loop system satisfies specification $\mathtt{CS}_1$ or $\mathtt{CS}_2$.
\end{prob}

We propose to solve Problem \ref{prob:1} by using a periodically switched controller based on a CLBF, as will be established in the next section. The CLBFs and controller modes are synthesized using grammar-guided genetic programming (introduced in Section \ref{sec:GP}) and verified by means of an SMT solver. The overall algorithm is described in Section \ref{sec:Alg}.

\section{Control strategy}
In this section we discuss the used control strategy and establish how it solves problem \ref{prob:1} by means of Theorem \ref{thm:1} and Corollary \ref{col:rsws}.

\subsection{Control Lyapunov Barrier Function}

Consider a set of controller modes with index set $Q \! \subset \! \mathbb{Z}_{\geq 0}$:
\begin{equation}
\mathcal{G} = \{g_q : \mathcal{X} \rightarrow \mathcal{U}  \mid q \in Q \}.
\label{eq:cmodes}
\end{equation}
Given the system \eqref{eq:system}, an initial state $x=\xi(t_k)$, let us denote the (over-approximated) reachable set for $t\in [t_k, t_k+h]$ under a controller mode $q$ as $R_q(x)$  s.t. given a $q$, $\forall t \in [t_k, t_k+h]: \xi(t) \in R_q(\xi(t_k))$. The construction of $R_q$ is discussed in Section \ref{sec:reach}. We consider a switching controller based on a CLBF defined as follows.

\begin{definition}[Control Lyapunov Barrier Function] A function $V \in \mathcal{C}^1(S, \R)$ is a Control Lyapunov Barrier Function (CLBF) w.r.t. the compact sets $(S,I,G)$, $S \subseteq  \mathcal{X}$, $I,G \subseteq int(S)$, system \eqref{eq:system}, and controller modes \eqref{eq:cmodes} if there exists a scalar $\gamma > 0$ such that
\begin{subequations}
\begin{align}
\forall x \in I : &~V(x) \leq 0 \label{eq:conI} \\
\forall x \in \partial S : & ~V(x) >  0  \label{eq:condS} \\
\forall x \in A \backslash G , \exists q \in Q, \forall z \in R_q(x):  &~ \dot{V}_q(x,z) \leq - \gamma \label{eq:conVd} 
\end{align}\label{eq:conV}
\end{subequations}
where $A:= \{ x\in S\mid  V(x) \leq 0\}$ and $ \dot{V}_q(x,z)= \langle \nabla V (z), f(z,g_q(x)) \rangle$.
\end{definition}

\begin{rem}
\label{rem:1}
The choice of $\gamma$ is arbitrary, because if a solution $V^*$ exists for $\gamma^*$, there always exists a linear transformation of $V^*$ such that the inequalities in \eqref{eq:conV} are satisfied for any $\gamma$.
\end{rem}

\begin{prop}
\label{prop:1}
Given a CLBF $V$, $\exists e \in \R$ s.t. $e = \inf_{x \in S \backslash G} V(x)$. Furthermore, the sublevel set $L_c:=\{x\in S \mid V(x) \leq c \}$ is compact. 
\end{prop}
\begin{proof}
Since $V(x)$ is continuous and $S$ is compact,  $V[S] \subset \R$ is compact and hence $V[S \backslash G] \subseteq V[S]$ is bounded, i.e.  $\exists e \in \R$ s.t. $e = \inf_{x \in S \backslash G} V(x)$. Moreover, $Y:= \{y \in V[S]~|~ y \leq c\}$ and its inverse image $V^{-1}[Y] = L_c$ are compact.
\end{proof}

\subsection{Control policy}

Given a CLBF $V$, we consider periodically switching controllers of the form
\begin{equation}
\left\{ \begin{array}{rl}
u(t) &= g_{q_k}(\xi(t_k)). \\
q_k (t_k)  &= \argmin \limits_{q\in Q} \max \limits_{z \in R_q(\xi(t_k)) }\dot{V}_q(\xi(t_k),z)
\end{array}\right.
\label{eq:controller}
\end{equation}
where $t_{k+1} = t_k +h$, $t \in [t_k, t_k+h)$.

\subsection{Reach while stay}

The presented controller strategy based on the CLBF enforces specification $\mathtt{CS}_1$, as shown in the following theorem.
\begin{theorem}
\label{thm:1}
Given a system \eqref{eq:system}, CLBF $V$ w.r.t. compact sets $(S,I,G)$ and controller \eqref{eq:controller}, then \eqref{eq:RWS} holds.
\end{theorem}

\begin{proof}
For $\xi(t_0) \in I$ it follows from \eqref{eq:conI} and the definition of $A$ that $V(\xi(t_0)) \in A$. From \eqref{eq:conVd} it follows that for all $\xi(t_k) \in A \backslash G$ there exists a $q \in Q$ such that $\forall t \in [t_k, t_k+h]: \dot{V}_q (\xi(t_k),\xi(t)) \leq -\gamma$. Selecting such a mode using controller \eqref{eq:controller}, applying the comparison theorem (see e.g. 
\cite{Khalil2002}), and using $\forall x \in A$, $V(x) \leq 0$, it follows that $\forall k \in \mathbb{Z}_{\geq 0}$, $\forall t \in [t_k, t_k+h]$, $\forall \xi(t_k) \in A \backslash G$:
$
V(\xi(t))  
 \leq V(\xi(t_k))- \gamma h \leq -\gamma h
$
Therefore, $\xi(t_k) \in A \backslash G$ implies $\forall t \in [t_k, t_k+h]$, $V(\xi(t))$ will decrease and thus cannot reach $\partial S$, as from \eqref{eq:condS} we have $\forall x \in \partial S:  V(x) > 0$. Since from proposition \ref{prop:1} it follows that $V(x)$ on $A \backslash G \subseteq S \backslash G $ is lower bounded, $V(\xi(t))$ will decrease until in finite time $\xi(t)$ leaves $A \backslash G$ and can only enter $G$, therefore \eqref{eq:RWS} holds.
\end{proof}

\subsection{Reach and stay while stay}

The conditions in \eqref{eq:conV} are not sufficient for forward invariance of (a subset of) the goal set, as they do not impose that $\forall x \in \partial G, V(x) < \inf_{y\in S \backslash G} V(y)$. Therefore some trajectories starting in $\partial G$ might enter $S\backslash G$ before entering $G$ again. The following corollary establishes sufficient conditions for specification $\mathtt{CS}_2$. 
\begin{corollary}
\label{col:rsws}
Given a system \eqref{eq:system}, CLBF $V$ w.r.t. compact sets $(S,I,G)$, and a controller \eqref{eq:controller}, if $\exists \beta \in \R$ such that
\begin{subequations}
\begin{align}
\forall x \in \partial G:~  V(x) > \beta \label{eq:condG} \\
 \! \!  \forall x \in G \backslash int(B), \exists q \in \! Q, \forall z \in \! R_q(x): 
\dot{V}_q(x,z) \leq - \gamma & \label{eq:conVd2}
\end{align}
 \label{eq:addV}
\end{subequations}
where $B:= \{x \in S \mid V(x) \leq \beta \}$, then \eqref{eq:RSWS} holds.
\end{corollary}
\begin{proof}
 From Theorem \ref{thm:1} we have that there exists a time $t_K \geq t_0$ such that $\xi(t_K) \in G$.
Analogous to the proof of Theorem \ref{thm:1}, from \eqref{eq:conVd2} it follows that $\forall \xi(t_K) \in G$, $\xi(t)$ with $t\geq t_K$ enters in finite time $G \cap B$. From the definition of $B$  and Proposition \ref{prop:1} it follows that $B$ is compact and thus $G \cap B$ is compact. From  \eqref{eq:conVd2} and controller \eqref{eq:controller} we have that $\forall x \in \partial (G \cap B), z \in R_q(x): \dot{V}_q(x,z) \leq  - \gamma$. Combining this with \eqref{eq:condG}, we have that all states $\xi(t) \in \partial (G \cap B)$ cannot reach $\partial G$ and $V(\xi(t))$ decreases, thus these trajectories will remain within $G \cap B$. Therefore it follows that $G \cap B \subseteq G$ is forward invariant.  As $G \subseteq int(S)$, we have that \eqref{eq:RSWS} holds. 
\end{proof}

\begin{rem}
Comparing the CLBF for RSWS to the CLBF in \cite{Verdier2017}, in this work the condition on the derivative of $V$ is only imposed for the sublevel set $A \subset S$, rather than the entire safe set $S$. Secondly, the CLBF in this work involves only 2 parameters $y$ and $\beta$, as opposed to 5 in  \cite{Verdier2017}.
\end{rem}

\section{Grammar-guided genetic programming}
\label{sec:GP}
Genetic programming (GP) is an evolutionary algorithm capable of synthesizing entire functions, in our case a CLBF and controller modes, that minimize a cost function, without pre-specifying a fixed structure \cite{Koza1992}. The algorithm is initialized with a random population of candidate solutions (individuals). Each individual is assigned a fitness using the fitness function, which reflects how well the design goal is satisfied. Individuals are then selected based on this fitness to undergo genetic operations. The resulting individuals form a new generation. This cycle is repeated for many generations under the expectation that the average fitness increases, until a solution is found or a maximum number of generations is met. 

In GP, solutions (the phenotypes) are encoded in a certain representation (the genotypes) that allows for easy modification. In this work we use a grammar-guided genetic programming (GGGP) algorithm, similar to the work of \cite{Whigham1995}, which uses a tree representation that is constructed based on a Backus-Naur form (BNF) grammar \cite{Backus1963}. The BNF grammar consists of the tuple $\{\mathcal{N},\mathcal{S},\mathcal{P},\mathcal{P}^* \}$, where $\mathcal{N}$ denotes the set of nonterminals, $\mathcal{S} \in \mathcal{N}$ the start symbol, $\mathcal{P}$ the set of production rules, and $\mathcal{P}^*$ the set of terminal production rules, which contains no recursive rules. An example of a simple grammar to construct monomials is given by $\mathcal{N} = \{ \mathrm{mon,var} \},~\mathcal{S}= \left< \mathrm{mon} \right>$, $\mathcal{P}$ in Table \ref{tab:PR1}, and $\mathcal{P}^*$ obtained by omitting the recursive rules from $\mathcal{P}$. Here $\rmm{mon}$ denotes monomials and $\rmm{var}$ scalar variables. Using $\mathcal{P}$, $\rmm{mon}$ can be mapped to either $\rmm{var}$ or $\rmm{var} \times \rmm{mon}$.

A parse tree is constructed using the BNF grammar as follows. Starting with the start symbol nonterminal $\mathcal{S}$, a random corresponding rule is chosen from the production rules $\mathcal{P}$. This rule forms a subtree that is put under the nonterminal. Subsequently, all nonterminals in the leaves of the resulting tree are similarly expanded, until all leave nodes contain no nonterminals anymore. To limit the tree depth, $\mathcal{P}^*$ is used if a predefined depth is reached, such that the number of recursive rules is limited. The final parse tree is transformed into the phenotype by replacing all nonterminals with their underlying subtrees, yielding a new parse tree corresponding directly to a function. Figure \ref{fig:GGGP} shows a fully grown genotype synthesized using the example grammar, as well as the transformation to its phenotype.

\begin{figure}%
\centering%
\begin{minipage}[t]{.5\textwidth}
\centering
\captionof{table}{Production rules $\mathcal{P}$}
\label{tab:PR1}
    \begin{tabular}{rl}
   $\mathcal{N}$&  Rules \\ \hline
$\left< \mathrm{mon} \right> $& $::= \rmm{var}$ \\
 & $~~~~~ |~\rmm{var} \times \rmm{mon} $  \\
 $\left< \mathrm{var} \right> $&  $::=a ~|~ b$ \\ \hline
    \end{tabular}%
\end{minipage}~~~~%
\begin{minipage}[t]{.5\textwidth}%
\centering%
\vspace{0pt}%
\includegraphics[scale=1]{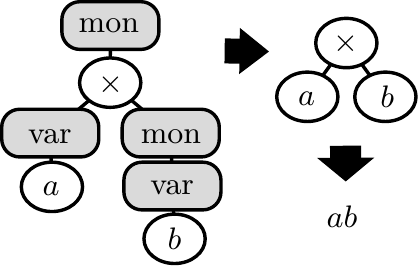}
\captionof{figure}{Genotype to phenotype}%
\label{fig:GGGP}%
\end{minipage}%
\end{figure}%

We use the genetic operators crossover and mutation, which take the role of exploitation and exploration of genotypes respectively. The crossover operator takes two individuals and switches two random subtrees with the same nonterminal root. The mutation operator takes a single individual and replaces a subtree corresponding to a random nonterminal with a new subtree grown from that nonterminal.

As stated before, we aim to synthesize the pair $(V,\mathcal{G})$. For both the CLBF and controller modes we use a separate parse tree, which we refer to as a gene. In this case, the genotype is formed by two genes.

\section{Automatic CLBF and controller synthesis}
\label{sec:Alg}
In this section the overall algorithm is described.

\subsection{One-step ahead reachable set}
\label{sec:reach}
In this work the reachable set is constructed by using Euler's forward method and bounding the local truncation error (LTE). This yields the following analytic expression
\begin{equation*}
r_q(x,\tau,e) = x+\tau f(x,g_q(x))+\frac{1}{2} \tau^2 e,
\end{equation*}
such that the over-approximated reachable set is given by 
\begin{equation}
R_q(s) =\bigcup_{(\tau,e) \in E} r_q(s,\tau,e)
\end{equation}
with $E := [0,h] \times \Pi_{i=1}^n[-\varepsilon_i,\varepsilon_i]$ and
\begin{equation}
\varepsilon_i = \max_{(x,u) \in \mathcal{X} \times \mathcal{U}} \left| \frac{\partial f_i(x,u)}{\partial x}f_i(x,u) \right|.
\end{equation}
While this construction can be quite conservative, it allows for relatively simple analytic expressions.

\subsection{Fitness}
For each inequality \eqref{eq:conI}-\eqref{eq:conVd} (and optionally \eqref{eq:condG}-\eqref{eq:conVd2}), an independent fitness value is constructed, consisting of a sample-based and an SMT solver-based fitness value. The former gives a measure of how much the inequalities in \eqref{eq:CLF1reform} are violated, whereas the SMT solver is used to provide a formal guarantee on whether these inequalities are satisfied. In this work we use the SMT solver dReal \cite{Gao2013}, which is able to verify nonlinear inequalities over the reals. Furthermore, in case a formula is not satisfied, the SMT solver can be used to provide a counterexample, which can again be used for the sample-based verification.

Inequalities \eqref{eq:conI}-\eqref{eq:conVd} and \eqref{eq:condG}-\eqref{eq:conVd2} can be rewritten as 
\footnote{By taking
$C_1 = I$, $C_2 = \partial S$, $C_3 = S \backslash G \times \Pi_{q \in Q} E$, $C_4 = \partial G$, $C_5 = G \times \Pi_{q \in Q} E$, 
$\phi_1(s) = - V(s)$, $\phi_2(s) =  V(s) -c$, $\phi_3(s) = \chi_A(s_x) (-\gamma -\min_{q\in Q}( \dot{V}(s_x,r_q(s_x,s_{\tau,q},s_{e,q}))$, $\phi_4(s) = V(s) -\beta -c$, $\phi_5(s) =  \chi_B(s_x)(-\gamma -\min_{q\in Q}( \dot{V}(s_x,r_q(s_x,s_{\tau,q},s_{e,q})))$, where for inequalities 3 and 5 the state $s$ is partitioned as $s = (s_x, s_{\tau,1}, s_{e,1}, \dots, s_{\tau,\bar{q}}, s_{e,\bar{q}})$ with $s_x \in \R^n$, $s_{\tau,q} \in \R$, $s_{e,q} \in \R^n$, $\bar{q} = |Q|$, $c$ is an arbitrary small real number to make the strict inequality non-strict, and $\chi_D(s)$ denotes a membership function of set $D$, i.e. $\chi_D(s) = 1$ if $s\in D$ and zero otherwise.}
\begin{equation}
\label{eq:CLF1reform}
(\forall s \in C_i) \phi_i(s) \geq 0, ~~~i = 1,\dots, 5.
\end{equation}
Given a finite set $C_{i,\mathrm{samp}} = \{x_1, \dots, x_n\}$, with $C_{i,\mathrm{samp}} \subset C_i$, the sample-based fitness is based on an error measure w.r.t. $\phi_i$ defined as
\begin{align*}%
e_{\phi_i}&:= \left\|  \left[ \min(0, \phi_i(x_1)), \dots, \min(0,\phi_i (x_n) \right] \right \|. 
\label{eq:errormeasures}
\end{align*}%
Using this measure, the sampled-based fitness is given by
\begin{equation}
f_{\mathrm{samp},\phi_i}: = ({1+e_{\phi_i}})^{-1}, ~~~i = 1,...,5.
\end{equation}

The SMT-based fitness $f_{\mathrm{SMT},\phi_i}$ is $1$ if it follows from dReal that the inequality is satisfied and $0$ otherwise. This fitness value is only computed if an individual satisfies $f_{\mathrm{samp},\phi_i} =1$ for all $i \in \{1,2,3\}$ (or $i \in \{1,...,5\}$). Otherwise, $f_{\mathrm{SMT},\phi_i}$ are set to 0 for all conditions.

To prioritize finding a $V$ that satisfies \eqref{eq:conI} and \eqref{eq:condS} before checking the condition on its derivative $\dot{V}$ in \eqref{eq:conVd}, (and similarly for the additional conditions \eqref{eq:condG}-\eqref{eq:conVd2}), we use the weights $w_i = \lfloor w_{i-1} f_{\mathrm{samp},\phi_{i-1}} \rfloor$ 
for $i = 2, \dots 5$ and $w_1 = 1$. The overall fitness is given by
\begin{equation}
f := \sum_{i=1}^j w_i f_{\mathrm{samp},\phi_i} + \sum_{i=1}^j  f_{\mathrm{SMT},\phi_i}, ~ j =3 \vee j=5. 
\label{eq:ffull}
\end{equation}

Finally, to promote the selection of equivalent, but less complex individuals, candidates with the same fitness \eqref{eq:ffull} are ranked according to the number of their parameters. If this is still not decisive, they are subsequently ranked based on their lowest maximum parameter. 

\subsection{Numerical optimization}
In order to speed up the convergence of the fitness, each generation the parameters of the individuals are optimized using Covariance Matrix Adaptation Evolution Strategy (CMA-ES) \cite{Hansen2001}, which is an evolutionary optimization algorithm that is regarded to be robust with regard to discontinuous fitness functions. We use the variant sep-CMA-ES \cite{Ros2008}, because of its linear time and space complexity. In our grammar, we have the rule `$\text{const}$' which creates a random constant. In every generation, these constants are then optimized using sep-CMA-ES, where their initial values are the current parameter values.

\subsection{Algorithm outline}

Provided a system, specification sets $(S,I,G)$, a grammar, and sample sets $C_{i,\mathrm{samp}}$, the proposed approach consists of the following steps:
\begin{enumerate}
\item A random population of individuals is generated as described in Section \ref{sec:GP}.
\item The parameters of the individuals are optimized using sep-CMA-ES based on the sample fitness.
\item The full fitness \eqref{eq:ffull} is computed.
\item Counterexamples returned by the SMT solver are added to $C_{i,\mathrm{samp}}$.
\item The best individuals are copied to the next generation.
\item A new population is created by selecting individuals using tournament selection \cite{Koza1992} and modifying them using genetic operators. 
\item Steps 3 to 5 are repeated until all conditions are satisfied (i.e. $f_{\mathrm{SMT},\phi_i}=1$ for all $i$) or a maximum number of generations is reached. 
\end{enumerate}
\begin{rem}
If the maximum number of generations is reached before an individual achieves $f_{\mathrm{SMT},\phi_i}=1$ for all $i$, no guarantees on the specification are provided.
\end{rem}
\begin{rem}
It is possible to pre-define the controller modes in $\mathcal{G}$, such that only the CLBF is synthesized.  
\end{rem}

\subsection{Additional operations}
To aid finding the correct bias $s$ of $V(x)$ such that \eqref{eq:conI} is satisfied, the following biasing is performed before each fitness evaluation within CMA-ES:
\begin{equation}
V'(x) = V(x) - \max \Big( \max_{x \in I_\mathrm{samp}}( V(x)), 0 \Big),
\end{equation}
where $ I_\mathrm{samp}$ denotes a subsampled set of $I$. 
To guide the search further,  we impose the additional condition
\begin{equation}
\forall x \in S\backslash G: V(x) \geq V(x_c)
\label{eq:addconstraint}
\end{equation}
where $x_c$ denotes the center of the goal set. 

\section{Implementation}
The switching law in \eqref{eq:controller} is computationally intensive to check online. By offline designing  $\alpha_q: \R^n \rightarrow \R$ for all $q \in Q$ such that
\begin{equation}
\begin{array}{r}
 \forall x \in D, \forall q \in Q:   \max \limits_{z \in R_q(x) }\dot{V}_q(x,z) > -\gamma \implies\\
  \min \limits_{p\in Q}  (\dot{V}_p (x,x)+\alpha_p(x)) < \dot{V}_q(x,x) +\alpha_q(x),
\end{array}
\label{eq:alphacon}
\end{equation}
allows us to replace the switching law with:
\begin{equation}
q_k(t_k)  = \argmin \limits_{q\in Q} ( \dot{V}_q(\xi(t_k),\xi(t_k))+\alpha_q(\xi(t_k)) ).
\label{eq:SWrelax}
\end{equation}
Intuitively, when at a point $x$ a mode $q'$ is not viable under the reachable set $R_q(x)$, the nominal system $\dot{V}_{q'}(x,x)$ plus buffer $\alpha_{q'}(x) $ should not minimize the set $\bigcup_{q\in Q} \dot{V}_q(x,x)+\alpha_q(x)$, such that it is not selected by the switching law.  
\begin{theorem}
\label{thm:rel}
Given a CLBF, if $\forall q \in Q$, $\alpha_q(x)$ satisfies \eqref{eq:alphacon} for $D =A\backslash G$, switching law \eqref{eq:SWrelax} yields that \eqref{eq:RWS} holds.
\end{theorem}
\begin{proof}
This proof is by contradiction. By definition of the CLBF, for all $x \in A\backslash G$, there always exists a $q$ such that $\max_{z \in R_{q}(x) }\dot{V}_{q} (x,z) \leq -\gamma $. Assume that when using switching law \eqref{eq:SWrelax}, we have $ \max_{z \in R_{q_k}(\xi(t)) }\dot{V}_{q_k} (\xi(t),z) >  -\gamma $. It then follows from \eqref{eq:alphacon} that $ \min_{q\in Q}  (\dot{V}_q (x,x)+\alpha_q(x)) < \dot{V}_{q_k}(x,x) +\alpha_{q_k}(x)$. This directly contradicts the switching law $q_k =  \min_{q\in Q}  (\dot{V}_q (x,x)+\alpha_q(x))$. Hence switching law \eqref{eq:SWrelax} can only select a $q_k$ such that $\max_{z \in R_{q_k}(\xi(t_k)) }\dot{V}_{q_k} (\xi(t_k),z) \leq -\gamma $, which is guaranteed to exist by the design of the CLBF. The remainder of the proof is analogous to the proof of Theorem \eqref{thm:1}.
\end{proof}
\begin{corollary}
Given a CLBF satisfying \eqref{eq:addV}, if $\forall q$, $\alpha_q(x)$ satisfies \eqref{eq:alphacon} for $D =A \backslash int(B)$, using switching law \eqref{eq:SWrelax} yields that \eqref{eq:RSWS} holds.
\end{corollary}
\begin{proof}
This proof is analogous to the proof of Theorem \ref{thm:rel} and Corollary \ref{col:rsws}.
\end{proof}
The functions $\alpha_q(x)$ can be designed and verified offline using again an SMT solver. 

\section{Case studies}
\begin{table*}[t]
\newcommand\Tstrut{\rule{0pt}{5.6ex}}         
\newcommand\Bstrut{\rule[-3.6ex]{0pt}{0pt}}   
\centering
\caption{Polynomial systems and results for 10 runs. $t$: total time, $\mu$: mean, $\sigma$: standard deviation.}
\begin{changemargin}{-2cm}{-2cm}
\label{tab:systems}
\scalebox{0.9}{
    \begin{tabular}{c|c|c|c}
  System & Linear & 2nd-order & 3rd-order \\
    \hline
    \Tstrut
$f(x,u)$ &  
$ \begin{array}{l}
\dot{x}_1 = x_2 \\ 
\dot{x}_2 = -x_1 +u
\end{array}
$ &
$
\begin{array}{ll}
\dot{x}_1 &= x_2-x_1^3\\ 
\dot{x}_2 &= u
\end{array} $
&
$\begin{array}{l}
\dot{x}_1=-10x_1 +10 x_2+u \\ 
\dot{x}_2 =28 x_1 - x_2-x_1 x_3 \\
\dot{x}_3 = x_1 x_2 - 2.6667 x_3
\end{array} $ \Bstrut \\
\hline 
\rule{0pt}{3.0ex}
$S$ &  $ [-1,~ 1]^2$ & $[-1,~ 1]^2$ & $[-5,~ 5]^2$ \\
$I$ & $ [-0.5, 0.5]^2$ & $[-0.5, 0.5]^2$ & $[-1.2, 1.2]^3$ \\
$G$ & $[-0.1, 0.1]^2$ & $[-0.05,~0.05]^2$ & $ [-0.3, 0.3]^3$ 

\rule[-1.6ex]{0pt}{0pt} 
\\
  \hline
\rule{0pt}{3.0ex}
$\mathcal{G}$ & $\{-1,0,1\}$&$\{-1,0,1\}$ & $\{-100,-50,-5,0,5,50,100\}$
\rule[-1.6ex]{0pt}{0pt} \\
\hline
\rule{0pt}{3.0ex}
$\varepsilon, h$
& 
$(2,1),~ 0.01 \mathrm{s}$
&
$(7, 0),~ 0.01 \mathrm{s}$
&
$(3800, 6800, 1900),~ 0.001 \mathrm{s}$

\rule[-1.6ex]{0pt}{0pt} \\
\hline
\hline

\rule{0pt}{2.0ex}
$\begin{array}{l}
~\\
\text{\# gen.} \\
t~ [\mathrm{s}]  \\
\end{array}
$
& 
$\begin{array}{cccc}
 \min & \max & \mu & \sigma\\
3	& 5	& 3.8	& 0.79 \\
11.55	& 21.34	& 15.67 &	3.21
\end{array}$
& 
$\begin{array}{cccc}
 \min & \max & \mu & \sigma\\
7 &	11&	9.1&	1.52 \\
32.26	&67.92&	47.81&	12.32
\end{array}$
&
$\begin{array}{cccc}
 \min & \max & \mu & \sigma\\
6	&50	&16.7 &16.03 \\
86.39	&524.96 &	205.08	&138.02
\end{array}$
\\
\hline
    \end{tabular}
    }
    \end{changemargin}
\end{table*}

\begin{table*}[t]
\newcommand\Tstrut{\rule{0pt}{6.6ex}}         
\newcommand\Bstrut{\rule[-5.6ex]{0pt}{0pt}}   
\centering
\caption{Nonpolynomial systems and results for 10 runs. $t$: total time, $\mu$: mean, $\sigma$: standard deviation}
\begin{changemargin}{-1.5cm}{-1.5cm}
\label{tab:systems2}
\scalebox{0.9}{
    \begin{tabular}{c|c|c}
  System & Pendulum  & Pendulum on a cart \\
    \hline
    \rule{0pt}{10.0ex}
$f(x,u)$&  
$\begin{array}{l}
 \dot{x}_1 = x_2 \\
 \dot{x}_2 =  - \left(\frac{b}{J}+\frac{K^2}{J R_a} \right) x_2  -\frac{m l g}{J}  \sin(x_1) +\frac{K}{J R_a}u\\
 ~ \\
  m =5.50\cdot 10^ {-2}  \textrm{ kg},~ l = 4.20 \cdot 10^ {-2} \textrm{ m},\\
 J = 1.91\cdot 10^{-4}\textrm{ kg} ~ \textrm{m}^2,~ g = 9.81 \textrm{ m} / \textrm{s}^2, \\K = 5.36\cdot 10^ {-2} \textrm{ Nm/A},~R_a = 9.50$ $\Omega, \\
 b= 3.0\cdot 10^{-6} \textrm{Nms}
\end{array} 
$
&
$\begin{array}{l}
\dot{x}_1 =x_2\\ 
\dot{x}_2 =\frac{g}{l}\sin(x_1) - \frac{b}{ml^2}x_2+\frac{1}{ml}\cos(x_1) u\\
 ~ \\
 g = 9.8 \textrm{ m}/ \textrm{s}^2 ,~ b = 2~\textrm{Nms} \\
 l = 0.5\textrm{ m}, ~ m =0.5\textrm{ kg}.
\end{array}$ 
   \rule{0pt}{10.0ex} \\
\hline 
\rule{0pt}{4.0ex}
$S$  & $\left[-2 \pi,~ 2\pi\right]\times \left[-100,~100\right]$ & $[-2\pi,2\pi]\times [-10,10]$ \\
$I$&  $\left[-\pi,~\pi\right]\times \left[-10,~10\right]$ & $[-0.5, 0.5]^2$ \\
$G$ & $\left[-1.0,~ -0.5\right]\times \left[-1.0,~1.0\right]$ & $[-0.25,0.25]\times [-1,1]$
\rule[-2.6ex]{0pt}{0pt} \\
  \hline
\rule{0pt}{3.0ex}
$\mathcal{G}$ &$\{-10,-5,0,5,10\}$ & $\{-6,-2,0,2,6\}$
\rule[-1.6ex]{0pt}{0pt} \\
\hline
\rule{0pt}{3.0ex}
$\varepsilon,~ h$
&
$(600, 12700),~ 0.001 \mathrm{s}$
&
$(200,3200),~  0.001 \mathrm{s}$
\rule[-1.6ex]{0pt}{0pt} \\
\hline
\hline
\rule{0pt}{2.0ex}
$\begin{array}{l}
~\\
\text{\# gen.} \\
t~ [\mathrm{s}]  \\
\end{array}
$
& 
$\begin{array}{cccc}
 \min & \max & \mu & \sigma\\
3&	12	&7.6& 	3.03 \\
32.97 &	185.96	&106.72& 	54.84

\end{array}$
&
$\begin{array}{cccc}
 \min & \max & \mu & \sigma\\
5	&16	&8.6&	3.47\\
48.02&	155.17&	85.2&	35.18
\end{array}$
\\
\hline
    \end{tabular}%
}\end{changemargin}
\end{table*}

In this section the effectiveness of the approach is demonstrated for a simple linear system, polynomial systems of second and third order (see Table \ref{tab:systems}), and two nonpolynomial systems (see Table \ref{tab:systems2}). The systems and specifications are adopted from \cite{ravanbakhsh2015ARCH} and references therein, with the exception of the Pendulum system, adopted from \cite{Verdier2017}. 

For these case studies, we fixed the control mode vector field $\mathcal{G}$ and synthesized controllers for the reach-while-stay specification $\mathtt{CS}_1$. Across all these case studies, we used a population of 16 individuals, a maximum of 50 generations, and a maximum of 30 generations within CMA-ES. The mutation and crossover rates were both chosen to be 0.5. The number of test samples  and maximum number of additional counterexamples were set to 100 and 300 respectively. For the counterexamples, a first-in-first-out principle was used. The (arbitrary) $\gamma$ of the CLBF was set to $\gamma =0.1$ and the precision parameter of dReal set to $\delta=0.001$. The values of $\varepsilon_i$ are obtained using bisection and the SMT solver. The GGGP algorithm and CMA-ES are implemented in Mathematica, running on an Intel Xeon CPU E5-1660 v3 3.00GHz using 8 CPU cores. 

The used grammar is defined by $\mathcal{S}_V = \rmm{const} + \rmm{expr}$, $\mathcal{N}$ and $\mathcal{P}$ as shown in Table \ref{tab:PRexamp}, and $\mathcal{P}^*$ is obtained by removing all recursive rules from $\mathcal{P}$. While this grammar restricts to polynomial CLBFs, the proposed approach can also be used for nonpolynomial CLBFs. Finally, the maximum recursive rule depth was set to be 7. 

\begin{table}[t]
\centering
\caption{Production rules $\mathcal{P}$ }
\label{tab:PRexamp}
\scalebox{1}{
    \begin{tabular}{rl}
    $\mathcal{N}$ & Rules \\ \hline
$\left< \mathrm{expr} \right> $ & $::=\rmm{expr} + \rmm{expr} ~|~ \rmm{pol} $\\
$\left< \mathrm{pol} \right> $ & $::=\rmm{pol} + \rmm{pol} ~|~ \rmm{const}\times\rmm{mon}  $\\
$\left< \mathrm{mon} \right>$ &$::= \rmm{var} ~|~ \rmm{var} \times \rmm{var} $ \\
 $\left< \mathrm{var} \right> $ & $::=x_1 -x_{c,1}~|~\dots ~|~  x_n-x_{c,n}$ \\
  $\left< \mathrm{const} \right> $ &$::=$ Random Real $\in \left[-10,10\right]$ \\
    $\left< \mathrm{G} \right>$ & $::=\{ \rmm{lin} \}~ |~ \dots ~|~  \{ \rmm{lin}, \rmm{lin},\rmm{lin} \}$,\\
  $\left< \mathrm{lin} \right> $ & $::=\rmm{const} (x_1-x_{c,1}) +\dots + \rmm{const} (x_n-x_{c,n}) $ \\ &$~~~~~|~\rmm{const} \rmm{var}~|~ \rmm{const} $\\
  \hline
    \end{tabular}%
    }
\end{table}

To show repeatability, the synthesis was repeated 10 times for each benchmark. Statistics on the number of generations and the total synthesis time are shown in Table \ref{tab:systems} and \ref{tab:systems2}. With the exception of the third-order polynomial system, in all 10 runs a solution was found for each benchmark. For the third-order system only a single run did not find a solution within 50 generations. 
Given the found solutions, we used again bisection and the SMT solver to find a $\beta$ such that the conditions in  Corollary \ref{col:rsws} hold. We found a $\beta$ such that the conditions in Corollary \ref{col:rsws} hold for: 4 solutions of the linear system, 1 of the 2nd-order system, 8 of the 3rd-order system, 0 of the pendulum system and 6 of the pendulum on cart system. Hence for these solutions the stronger specification $\mathtt{CS}_2$ is guaranteed.

One of the found solutions for the pendulum system is
\begin{equation*}
V(x) \!=\! -4015.83 + 10.8526 x_1' + 199.048 x_1'^2 + 
 0.311673 x_2 + 18.8116 x_1' x_2 + 2.23916 x_2^2,
\label{eq:Vpend}
\end{equation*}
where $x_1'=(0.75 + x_1)$.
We manually designed $\alpha(x) = [\alpha_1(x), \dots \alpha_5(x)]^T $ to be $[100,0,0,0,500]$ such that \eqref{eq:alphacon} holds.  Figure \ref{fig:pendSim} shows the phase plot of the closed-loop system for $\xi(t_0) \in \{(-\pi,10),(-2,-5),(1.5,0),(\pi,10)$\}. It can be seen that indeed all trajectories satisfy $\mathtt{CS}_1$.

For this solution, we could not find a $\beta$ such that Corollary \ref{col:rsws} holds. Nevertheless, by increasing the goal set to $G$ to $[-1,-0.5] \times [-1.5, 1.5]$, it can be verified that for $\beta = -4012.3$ the conditions in Corollary \ref{col:rsws} hold.  

\begin{figure}
\centering
\includegraphics[scale =0.9]{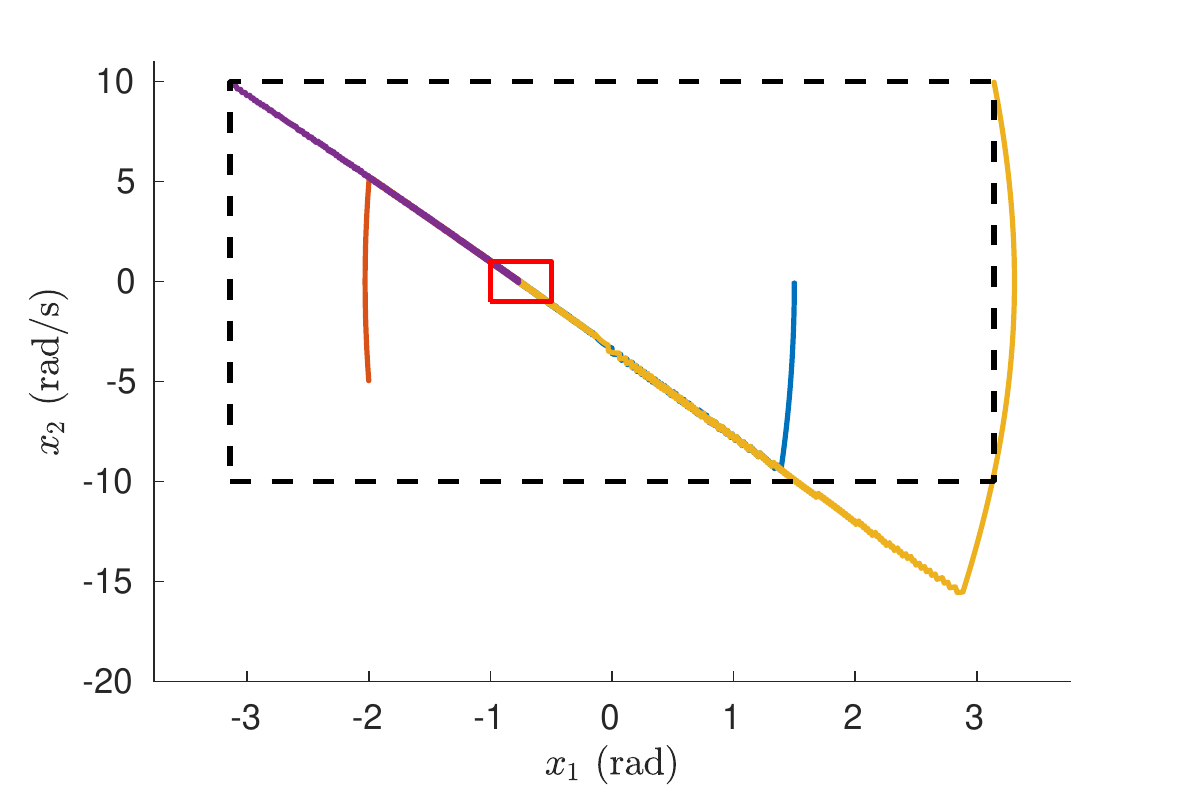}
\caption{Phase diagram of different initial conditions for the pendulum system using a found CLBF. Dashed: initial set, red: goal set.}
\label{fig:pendSim}
\end{figure}

The used sampling times $h$ are significant larger than the minimal dwell-times reported in \cite{Verdier2017} and \cite{ravanbakhsh2016}. For example, for the pendulum on cart system benchmark we used $h = 0.001$ seconds, whereas \cite{ravanbakhsh2016} reports a theoretical minimum dwell-time of $2 \cdot 10^{-6}$ seconds.

\subsection{Evolving $\mathcal{G}$}
Let us reconsider the pendulum on a cart from Table \ref{tab:systems2} and specification $\mathtt{CS}_1$, but without pre-specifying $\mathcal{G}$. We saturate the input with 
\begin{equation*}
u(t_k) = \max(-6, \min(6, g_{q_k}(\xi(t_k)))).
\end{equation*}
A separate gene for the controller modes $\mathcal{G}$ is used with start symbol $\mathcal{S}_G = \rmm{G}$ and the product rules in Table \ref{tab:PRexamp}. The results for 10 runs are shown in Table \ref{tab:CPexp1}. Comparing Table \ref{tab:systems2} with \ref{tab:CPexp1} we observe a comparable number of generations required to find a solution, although a longer computation time per generation is observed. However, the benefit is that no discretization of the input space is required. One of the found solutions is given by
\begin{align*}
V &=  -22.2281+52.1542 x_1^2+13.0965 x_1 x_2+17.3873 x_2^2,\\
\mathcal{G}&= \{-11.0824 x_1-13.2558 x_2\} .
\end{align*}
Note that $\mathcal{G}$ consists of only a single mode, hence no switching law is required when implementing this controller. 
Finally, for $\beta =-19.5313$, $V$ satisfies \eqref{eq:addV}, hence using this controller also guarantees $\mathtt{CS}_2$. 

\begin{table}[t]%
\caption{Results for 10 runs for the Pendulum on a cart system without pre-defining $\mathcal{G}$. }%
    \centering%
\label{tab:CPexp1}%
    \begin{tabular}{rcccc}%
    \hline%
   &   Min & Max & $\mu$ & $\sigma$ \\%
   \hline %
  \# gen. & 4	&14	&7.6	&2.84 \\%
   $t ~[\mathrm{s}]$  & 135.59&	536.12	&267.82	&124.48%
    \end{tabular}%
\end{table}%

\section{Discussion}

This paper presented a method for automatic synthesis of a periodically switched state feedback controller for nonlinear sampled-data systems with reachability and safety specifications. Preliminary results have been shown for several nonlinear systems up to the third-order. It was shown that the framework was able to synthesize CLBFs given pre-defined controller modes, but was also capable of synthesizing controller modes automatically, eliminating the need to discretize the input space beforehand. Moreover, it is possible to find a single controller automatically, removing the need for the switching law entirely. 

For almost all benchmarks runs, solutions were found within 50 generations. Nevertheless, a drawback of the proposed methodology is that there is no guarantee that a solution will be found within a number of generations.

A straightforward improvement to obtain less conservative sampling times is by using a less conservative over-approximation of the reachable set, e.g. by using higher order Taylor series approximations or using local bounds rather than for the entire domain.

Finally, the functions $\alpha_q(x)$ that simplify the switching condition are currently synthesized by hand. In future work, the aim is to automate this synthesis as well, for example by again using the combination of GP with SMT solvers.

\bibliographystyle{IEEEtran}

\bibliography{IEEEabrv,bib}

\end{document}